\documentclass[a4paper]{easychair}
\usepackage{submission}  
\begin{document}

\title{A \tlaplus Proof System}

\titlerunning{A \tlaplus Proof System}


\author{
  Kaustuv Chaudhuri \\
  INRIA \\
  \and
  Damien Doligez \\
  INRIA \\
  \and
  Leslie Lamport \\
  Microsoft Research \\
  \and
  Stephan Merz \\
  INRIA \& Loria
}

\authorrunning{Chaudhuri, Doligez, Lamport, and Merz}

\maketitle

\begin{abstract}
  We describe an extension to the \tlaplus specification language
  with constructs for writing proofs and a proof environment, called
  the Proof Manager (PM), to checks those proofs.  The language and
  the \PM support the incremental development and checking of
  hierarchically structured proofs.  The \PM translates a proof into
  a set of independent proof obligations and calls upon a collection
  of back-end provers to verify them.  Different provers can be used
  to verify different obligations.  The currently supported back-ends
  are the tableau prover Zenon and Isabelle/\tlaplus, an
  axiomatisation of \tlaplus in Isabelle/Pure.  The proof obligations
  for a complete \tlatwo proof can also be used to certify the theorem
  in Isabelle/\tlaplus.
\end{abstract}

\section{Introduction}
\label{sec:intro}

\tlaplus is a language for specifying the behavior of concurrent and
distributed systems and asserting properties of those
systems~\cite{lamport03tla}.  However, it provides no way to write
proofs of those properties.  We have designed an extended version of
the language that allows writing proofs, and we have begun
implementing a system centered around a \textit{Proof Manager} (\PM)
that invokes existing automated and interactive proof systems to check
those proofs. For now, the new version of \tlaplus is called \tlatwo
to distinguish it from the current one.  We describe here the \tlatwo
proof constructs and the current state of the proof system.

The primary goal of \tlatwo and the proof system is the mechanical
verification of systems specifications. The proof system must not only
support the modal and temporal aspects of TLA needed
to reason about system properties, but must also support
ordinary mathematical reasoning in the underlying logic. Proofs in
\tlatwo are natural deduction proofs written in a hierarchical style
that we have found to be good for ordinary
mathematics~\cite{lamport93amm} and crucial for managing the
complexity of correctness proofs of systems~\cite{gafni:disk-paxos}.

The \PM computes proof obligations that establish the correctness of
the proof and sends them to one or more back-end provers to be
verified.  Currently, the back-end provers are Isabelle/\tlaplus, a
faithful axiomatization of \tlaplus in Isabelle/Pure, and
Zenon~\cite{bonichon07lpar}, a tableau prover for classical
first-order logic with equality.  The \PM first sends a proof
obligation to Zenon.  If Zenon succeeds, it produces an Isar script
that the \PM sends to Isabelle to check.
Otherwise, the \PM outputs an Isar script that uses one of Isabelle's
automated tactics.
In both cases, the obligations are certified by Isabelle/\tlaplus.
The system architecture easily accommodates other back-end provers; if
these are proof-producing, then we can use their proofs to certify the
obligations in Isabelle/\tlaplus, resulting in high confidence in the
overall correctness of the proof.

The \tlatwo proof constructs are described in
Section~\ref{sec:proof-language}.  Section~\ref{sec:obligations}
describes the proof obligations generated by the \PM, and
Section~\ref{sec:backend} describes how the \PM uses Zenon and
Isabelle to verify them.  The conclusion summarizes what we have done
and not yet done and briefly discusses related work.

\section{\tlaplus and its Proof Language}
\label{sec:proof-language}

\subsection{TLA}
\label{sec:proof-language.tla} 

The \tlaplus language is based on the Temporal Logic of Actions
(TLA)~\cite{lamport:newtla}, a linear-time temporal logic. The rigid
variables of TLA are called \emph{constants} and the flexible
variables are called simply \emph{variables}.  TLA assumes an
underlying ordinary (non-modal) logic for constructing expressions.
Operators of that logic are called \emph{constant} operators.  A
\emph{state function} is an expression built from constant operators
and TLA constants and variables.  The elementary (non-temporal)
formulas of TLA are \textit{actions}, which are formulas built with
constant operators, constants, variables, and expressions of the form
$f'$, where $f$ is a state function.  (TLA also has an \ENABLED
operator that is used in expressing fairness, but we ignore it
for brevity.)  An action is interpreted as a predicate on pairs of
states that describes a set of possible state transitions, where state
functions refer to the starting state and primed state functions refer
to the ending state.  Because priming distributes over constant
operators and because $c'$ is equal to $c$ for any constant $c$, an
action can be reduced to a formula built from constant operators,
constants, variables, and primed variables.

TLA is practical for describing systems because all the complexity of
a specification is in the action formulas.  Temporal operators are
essentially used only to assert liveness properties, including
fairness of system actions.  Most of the work in a TLA proof is in
proving action formulas; temporal reasoning occurs only in proving
liveness properties and is limited to propositional temporal logic and
to applying a handful of proof rules whose main premises are action
formulas.  Because temporal reasoning is such a small part of TLA
proofs, we have deferred its implementation.  The \PM now handles
only action formulas.  We have enough experience mechanizing TLA's
temporal reasoning~\cite{engberg:mechanical} to be fairly confident
that it will not be hard to extend the \PM to support it.

A formula built from constant operators, constants, variables, and
primed variables is valid iff it is a valid formula of the underlying
logic when constants, variables, and primed variables are treated as
distinct variables of the logic---that is, if $v$ and $v'$ are
considered to be two distinct variables of the underlying logic, for
any TLA variable $v$.  Since any action formula is reducible to such a
formula, action reasoning is immediately reducible to reasoning in the
underlying logic.  We therefore ignore variables and priming here and
consider only constant formulas.

\subsection{\tlaplus}

The \tlaplus language adds the following to the TLA logic:
\begin{icom}
\item An underlying logic that is essentially ZFC set theory plus
  classical untyped first-order logic with Hilbert's
  $\varepsilon$~\cite{leisenring:mathematical-logic}.  The major
  difference between this underlying logic and traditional ZFC is that
  functions are defined axiomatically rather than being represented as
  sets of ordered pairs.

\item A mechanism for defining operators, where a user-defined
  operator is essentially a macro that is expanded syntactically.
  (\tlaplus permits recursive function definitions, but they are
  translated to ordinary definitions using Hilbert's $\varepsilon$.)

\item Modules, where one module can import definitions
  and theorems from other modules.  A module is parameterized by its
  declared variables and constants, and it may be instantiated in another
  module by substituting expressions for its parameters. The
  combination of substitution and the \ENABLED\ operator introduces
  some complications, but space limitations prevent us from discussing
  them, so we largely ignore modules in this paper.
\end{icom}
\tlaplus has been extensively documented~\cite{lamport03tla}.  Since
we are concerned only with reasoning about its underlying logic, which
is a very familiar one, we do not bother to describe \tlaplus in any
detail.  All of its nonstandard notation that appears in our examples is
explained.

\subsection{The Proof Language}
\label{sec:proof-language.lang}

The major new feature of \tlatwo is its proof language.  (For reasons
having nothing to do with proofs, \tlatwo also introduces recursive
operator definitions, which we ignore here for brevity.)  We describe
the basic proof language, omitting a few constructs
that concern aspects such as module instantiation that we are not
discussing.  \tlatwo also adds constructs for naming subexpressions
of a definition or theorem, which is important in practice for writing
proofs but is orthogonal to the concerns of this paper.

The goal of the language is to make proofs easy to read and write for
someone with no knowledge of how the proofs are being checked.  This
leads to a mostly declarative language, built around the uses and
proofs of assertions rather than around the application of
proof-search tactics.  It is therefore more akin to
Isabelle/Isar~\cite{isar} than to more operational interactive
languages such as Coq's Vernacular~\cite{coq}.
Nevertheless, the proof language does include a few operational
constructs that can eliminate the repetition of common idioms, albeit
with some loss of perspicuity.

At any point in a \tlaplus proof, there is a current obligation that
is to be proved.  The obligation contains a \emph{context} of known
facts, definitions, and declarations, and a \emph{goal}.
The obligation claims that the goal is logically entailed by the
context.  Some of the facts and definitions in the context are marked
(explicitly or implicitly) as \emph{usable} for reasoning, while the
remaining facts and definitions are \textit{hidden}.

Proofs are structured hierarchically. The leaf (lowest-level) proof
\OBVIOUS\ asserts that the current goal follows easily from the usable
facts and definitions.  The leaf proof
\begin{gather*}
  \BY\ e_{1},\ldots, e_{m} \ \DEFS\ o_{1},\ldots, o_{n}
\end{gather*}
asserts that the current goal follows easily from the usable facts and
definitions together with (i)~the facts $e_{i}$ that must themselves
follow easily from the context and (ii)~the known definitions of
$o_{j}$.  Whether a goal follows easily from definitions and facts
depends on who is trying to prove it.  For each leaf proof,
the \PM sends the corresponding \emph{leaf obligation} 
to the back-end provers, so in practice ``follows easily''
means that a back-end prover can prove it.
A non-leaf proof is a sequence of \textit{steps}, each consisting of a
begin-step token and a proof construct.  For some constructs
(including a simple assertion of a proposition) the step takes a
subproof, which may be omitted.  The final step in the sequence simply
asserts the current goal, which is represented by the token \QED.
A begin-step token is either a \emph{level token} of the form \s{n} or
a \emph{label} of the form \s{n}"l", where $n$ is a level number that
is the same for all steps of this non-leaf proof, and "l" is an
arbitrary name.  The hierarchical structure is deduced from the level
numbers of the begin-step tokens, a higher level number beginning a
subproof.

Some steps make declarations or definitions or change the current goal
and do not require a proof.  Other steps make assertions that become
the current goals for their proofs.  An omitted proof (or one
consisting of the token \OMITTED) is considered to be a leaf proof
that instructs the assertion to be accepted as true.  Of course, the
proof is then incomplete.  From a logical point of view, an omitted
step is the same as an additional assumption added to the theorem;
from a practical point of view, it doesn't have to be lifted from its
context and stated at the start.  Omitted steps are intended to be
used only in the intermediate stages of writing a proof.

Following a step that makes an assertion (and the step's proof), until
the end of the current proof (after the \QED\ step), the contexts
contain that assertion in their sets of known facts.  The assertion is
marked usable iff the begin-step token is a level token; otherwise it
can be referred to by its label in a \BY\ proof or made usable with
a \USE\ step.

The hierarchical structure of proofs not only aids in reading the finished
proof but is also quite useful in incrementally writing proofs.  The
steps of a non-leaf proof are first written with all proofs but that
of the \QED\ step omitted.  After checking the proof of the \QED
step, the proofs omitted for other steps in this or earlier levels
are written in any order.  When writing the proof, one may discover
facts that are needed in the proofs of multiple steps.
Such a fact is then added to the proof as an earlier step, or added at a
higher level.  It can also be removed from the proof of the theorem
and proved separately as a lemma.  However, the hierarchical proof
language encourages facts relevant only for a particular proof to be
kept within the proof, making the proof's structure easier to see and
simplifying maintenance of the proof.  For correctness proofs of
systems, the first few levels of the hierarchy are generally
determined by the structure of the formula to be proved---for example,
the proof that a formula implies a conjunction usually consists of steps
asserting that it implies each conjunct.

As an example, we incrementally construct a hierarchical proof of
Cantor's theorem, which states that there is no surjective function
from a set to its powerset. It is written in \tlaplus as:
\begin{quote} \small
  \begin{tabbing}
    \THEOREM\ "\forall S : \forall f \in [S -> \SUBSET\ S] : 
        \exists A \in \SUBSET\ S : \forall x \in S : f[x] \neq A"
  \end{tabbing}
\end{quote}
where function application is written using square brackets,
"\SUBSET\ S" is the powerset of "S", and "[S -> T]" is the set of
functions from $S$ to $T$.

The statement of the theorem is the current goal for its top-level
proof. A goal of the form $\forall v:e$ is proved by introducing a
generic constant and proving the formula obtained by substituting it
for the bound identifier. We express this as follows, using the
\ASSUME/\PROVE construct of \tlatwo:
\begin{quote} \small
  \begin{tabbing}
    \THEOREM\ "\forall S : \forall f \in [S -> \SUBSET\ S] : 
                \exists A \in \SUBSET\ S : \forall x \in S : f[x] \neq A" \\
    \LSP \= \s11.\ \= \ASSUME \= "\NEW\ S", \\
         \>        \>         \> "\NEW\ f \in [S -> \SUBSET\ S]"\\
         \>        \> \PROVE "\exists A \in \SUBSET\ S : \forall x \in S : f[x] \neq A" \\
         \> \s12.  \> \QED \BY \s11
  \end{tabbing}
\end{quote}
Although we could have used labels such as \s1"one" and \s1"last"
instead of \s11 and \s12, we have found that proofs are easier to read
when steps at the same level are labeled with consecutive numbers.
One typically starts using consecutive step numbers and then uses
labels like \s32a for inserting additional steps.  When the proof is
finished, steps are renumbered consecutively.  (A planned user
interface will automate this renumbering.)

Step \s11 asserts that for any constants "S" and "f" with "f \in [S ->
\SUBSET\ S]", the proposition to the right of the \PROVE is true.
More precisely, the current context for the (as yet unwritten) proof
of \s11 contains the declarations of $S$ and $f$ and the usable fact
"f \in [S -> \SUBSET\ S]", and the \PROVE\ assertion is its goal.  The
\QED step states that the original goal (the theorem) follows from the
assertion in step~\s11.

We tell the \PM to check this (incomplete) proof, which it does by
having the back-end provers verify the proof obligation for the \QED
step.  The verification succeeds, and we now continue by writing the
proof of \s11.  (Had the verification failed because \s11 did not
imply the current goal, we would have caught the error before
attempting to prove \s11, which we expect to be harder to do.)

We optimistically start with the proof \OBVIOUS, but it is too hard
for the back-end to prove, and the \PM reports a timeout.  Often this
means that a necessary fact or definition in the context is hidden and
we merely have to make it usable with a \USE step or a \BY proof.  In
this case we have no such hidden assumptions, so we must refine the
goal into simpler goals with a non-leaf proof.  We let this proof have
level 2 (we can use any level greater than 1).  Since the goal itself
is existentially quantified, we must supply a witness.  In this case,
the witness is the classic diagonal set, which we call~"T".
\begin{quote} \small
  \begin{tabbing}
    \THEOREM\ "\forall S : \forall f \in [S -> \SUBSET\ S] : \exists A \in \SUBSET\ S : \forall x \in S : f[x] \neq A" \kill
    \PROOF \kill
    \LSP \= \s11.\ \= \ASSUME \= "\NEW\ S", \\
         \>        \>         \> "\NEW\ f \in [S -> \SUBSET\ S]" \\
         \>        \> \PROVE "\exists A \in \SUBSET\ S : \forall x \in S : f[x] \neq A" \\
         \>   \hspace{1em}     \= \s21.\ \= \DEFINE "T \DEF \{z \in S : z \notin f[z]\}" \\
         \>        \> \s22.  \> "\forall x \in S : f[x] \neq T" \\
         \>        \> \s23.  \> \QED \BY \s22
  \end{tabbing}
\end{quote}
Because definitions made within a proof are usable by default, the
definition of $T$ is usable in the proofs of \s22 and \s23.  Once
again, the proof of the \QED\ step is automatically verified, so all
that remains is to prove \s22.  (The \DEFINE\ step requires no proof.)

The system accepts \OBVIOUS\ as the proof of \s22 because the only
difficulty in the proof of \s11 is finding the witness. However,
suppose we want to add another level of proof for the benefit of a
human reader.  The universal quantification is proved as above, by
introducing a fresh constant:

\begin{quote} \small
  \begin{tabbing}
    \THEOREM\ "\forall S : \forall f \in [S -> \SUBSET\ S] : \exists A \in \SUBSET\ S : \forall x \in S : f[x] \neq A" \kill
    \PROOF \kill
    \LSP \= \s11.\ \= \ASSUME \= "\NEW\ S", \kill
         \>        \>         \> "\NEW\ f \in [S -> \SUBSET\ S]" \kill
         \>        \> \PROVE "\exists A \in \SUBSET\ S : \forall x \in S : f[x] \neq A" \kill
         \>        \> \PROOF \kill
         \>   \hspace{1em} \= \s21.\ \= \DEFINE "T == \{z \in S : z \notin f[z]\}" \kill
         \>        \> \s22.  \> "\forall x \in S : f[x] \neq T" \\
         \>        \> \hspace{1em} \= \s31.\ \= \ASSUME "\NEW\ x \in S" \PROVE "f[x] \neq T" \\
         \>        \>        \> \s32.\ \> \QED \BY \s31
  \end{tabbing}
\end{quote}
Naturally, the \QED step is verified.  Although the system accepts
\OBVIOUS\ as the proof of \s31 (remember that it could verify \s22 by
itself), we can provide more detail with yet another level
of proof.  We write this proof the way it would seem natural to a
person---by breaking it into two cases:
\begin{quote} \small
  \begin{tabbing}
    \THEOREM\ "\forall S : \forall f \in [S -> \SUBSET\ S] : \exists A \in \SUBSET\ S : \forall x \in S : f[x] \neq A" \kill
    \PROOF \kill
    \LSP \= \s11.\ \= \ASSUME \= "\NEW\ S", \kill
         \>        \>         \> "\NEW\ f \in [S -> \SUBSET\ S]" \kill
         \>        \> \PROVE "\exists A \in \SUBSET\ S : \forall x \in S : f[x] \neq A" \kill
         \>        \> \PROOF \kill
         \> \hspace{1em} \= \s21.\ \= \DEFINE "T == \{z \in S : z \notin f[z]\}" \kill
         \>        \> \s22.  \> "\forall x \in S : f[x] \neq T" \kill
         \>        \>  \hspace{1em} \= \s31.\ \= \ASSUME "\NEW\ x \in S" \PROVE "f[x] \neq T" \\
         \>        \>        \> \hspace{1em} \= \s41.\ \= \CASE "x \in T" \\
         \>        \>        \>        \> \s42.\ \> \CASE "x \notin T" \\
         \>        \>        \>        \> \s43.\ \> \QED \BY \s41, \s42
  \end{tabbing}
\end{quote}
The (omitted) proof of the \CASE\ statement \s41 has as its goal
"f[x]\neq T" and has the additional usable fact $x\in T$ in its context.

We continue refining the proof in this way, stopping 
with an \OBVIOUS or \BY proof when a goal is
obvious enough for the back-end prover or for a human reader,
depending on who the proof is being written for. 
A \BY\ statement can guide the prover or the human reader
by listing helpful obvious consequences of known facts.  
For example, the proof of \s41 might be "\BY\ x \notin f[x]".
The proof is now finished: it contains no omitted sub-proofs.  For
reference, the complete text of the proof is given in
Appendix~\ref{apx:cantor}.

Our experience writing hand proofs makes us expect that proofs of
systems could be ten or more levels deep, with the first several
levels dictated by the structure of the property to be proved.  
Our method of numbering steps makes such proofs manageable, and we are
not aware of any good alternative.

This example illustrates how the proof language supports the
hierarchical, non-linear, and incremental development of proofs.  The
proof writer can work on the most problematic unproved steps first,
leaving the easier ones for later.  Finding that a step cannot be
proved (for example, because it is invalid) may require changing other
steps, making proofs of those other steps wasted effort.  We intend to
provide an interface to the \PM that will make it easy for the user
to indicate which proofs should be checked and will avoid
unnecessarily rechecking proofs.

The example also shows how already-proved facts are generally not made
usable, but are invoked explicitly in \BY\ proofs.  Global definitions
are also hidden by default and the user must explicitly make them
usable.  This makes proofs easier to read by telling the reader what
facts and definitions are being used to prove each step.  It also
helps constrain the search space for an automated back-end prover,
leading to more efficient verification.  Facts and definitions can be
switched between usable and hidden by \USE\ and \HIDE\ steps, which
have the same syntax as \BY. As noted above, omitting the label from a
step's starting token (for example, writing \s4 instead of \s42) makes
the fact it asserts usable.  This might be done for compactness at
the lowest levels of a proof.

The example also indicates how the current proof obligation at every
step of the proof is clear, having been written explicitly in a parent
assertion.  This clear structure comes at the cost of introducing many
levels of proof, which can be inconvenient.  One way of avoiding these
extra levels is by using an assertion of the form "\SUFFICES\ A",
which asserts that proving $A$ proves the current goal, and makes $A$
the new current goal in subsequent steps.  In our example proof, one
level in the proof of step \s22 can be eliminated by writing the proof
as:
\begin{quote} \small
  \begin{tabbing}
    \THEOREM\ "\forall S : \forall f \in [S -> \SUBSET\ S] : \exists A \in \SUBSET\ S : \forall x \in S : f[x] \neq A" \kill
    \PROOF \kill
    \LSP \= \s11.\ \= \ASSUME \= "\NEW\ S", \kill
         \>        \>         \> "\NEW\ f \in [S -> \SUBSET\ S]" \kill
         \>        \> \PROVE "\exists A \in \SUBSET\ S : \forall x \in S : f[x] \neq A" \kill
         \>        \> \PROOF \kill
         \> \hspace{1em} \= \s21.\ \= \DEFINE "T == \{z \in S : z \notin f[z]\}" \kill
         \>        \> \s22.  \> "\forall x \in S : f[x] \neq T" \\
         \>        \>  \hspace{1em} \= \s31.\ \= \SUFFICES \ASSUME "\NEW\ x \in S" \PROVE "f[x] \neq T" \\
         \>        \>        \> \hspace{1em} \PROOF \OBVIOUS \\
         \>        \>        \> \s32.\ \> \CASE "x \in T" \\
         \>        \>        \> \s33.\ \> \CASE "x \notin T" \\
         \>        \>        \> \s34.\ \> \QED \BY \s32, \s33
  \end{tabbing}
\end{quote}
where the proofs of the \CASE\ steps are the same as before.  The
\SUFFICES\ statement changes the current goal of the level-3 proof to
$f[x]\neq T$ after adding a declaration of "x" and the usable fact "x
\in S" to the context. This way of proving a universally quantified
formula is sufficiently common that \tlatwo provides a \TAKE\
construct that allows the \SUFFICES\ assertion \s31 and its \OBVIOUS
proof to be written \mbox{\,$\TAKE\ x \in S$\,}. 

There is a similar construct, "\WITNESS\ f \in S" for proving an
existentially quantified goal $\exists x\in S: e$, which changes the
goal to "e[x := f]".
For implicational goals "e => f", the construct $\HAVE\ e$ changes the
goal to $f$.  No other constructs in the \tlatwo proof language change
the form of the current goal. We advise that these constructs be used
only at the lowest levels of the proof, since the new goal they create
must be derived instead of being available textually in a parent
assertion.  (As a check and an aid to the reader, one can at any point
insert a redundant \SUFFICES\ step that simply asserts the current
goal.)

The final \tlatwo proof construct is $\PICK\ x : e$, which introduces
a new symbol $x$ that satisfies $e$.  The goal of the proof of this
\PICK step is $\exists x : e$, and it changes the context of
subsequent steps by adding a declaration of "x" and the fact "e". 
A more formal summary of the language appears in Appendix~\ref{apx}.

The semantics of a \tlatwo proof is independent of any back-end
prover. Different provers will have different notions of what
``follows easily'', so an \OBVIOUS\ proof may be verified by one
prover and not another.  In practice, many provers such as Isabelle
must be directed to use decision procedures or special tactics to
prove some assertions.  For this purpose, special standard modules
will contain dummy theorems for giving directives to the
\PM.  Using such a theorem (with a \USE\ step or \BY\ proof) will
cause the \PM not to use it as a fact, but instead to generate
special directives for back-end provers.  It could even cause the \PM
to use a different back-end prover.  (If possible, the dummy theorem
will assert a true fact that suggests the purpose of the directive.)
\ednote{}{}For instance, using the theorem \emph{Arithmetic}
might be interpreted as an instruction to use a decision procedure for
integers.
We hope that almost all uses of this feature will leave the \tlatwo
proof independent of the back-end provers.  The proof will not have
to be changed if the \PM is reconfigured to replace one decision
procedure with a different one.

\section{Proof Obligations}
\label{sec:obligations}

The \PM generates a separate \textit{proof obligation} for each leaf
proof and orchestrates the back-end provers to verify these
obligations.  Each obligation is independent and can be proved
individually.  If the system cannot verify an obligation within a
reasonable amount of time, the \PM reports a failure.  The user must
then determine if it failed because it depends on hidden facts or
definitions, or if the goal is too complex and needs to be refined
with another level of proof.  (Hiding facts or definitions might also
help to constrain the search space of the back-end provers.)

When the back-end provers fail to find a proof, the user will know
which obligation failed---that is, she will be told the
obligation's usable context and goal and the leaf proof from which it
was generated.  We do not yet know if this will be sufficient in
practice or if the \PM will need to provide the user with more
information about why an obligation failed.  For example, many SAT and
SMT solvers produce counterexamples for an unprovable formula that
can provide useful debugging information.

The \PM will also mediate the \textit{certification} of the \tlatwo
theorem in a formal axiomatization of \tlatwo in a trusted logical
framework, which in the current design is Isabelle/\tlaplus (described
in Section~\ref{sec:backend.isa}). Although the \PM is designed
generically and can support other similar frameworks, for the rest of
this paper we will limit our attention to Isabelle/\tlaplus.
Assuming that Isabelle/\tlaplus is sound, once it has certified a
theorem we know that an error is possible only if the \PM incorrectly
translated the statement of the theorem into Isabelle/\tlaplus.

After certifying the proof obligations generated for the leaf proofs,
called the \textit{leaf obligations}, certification of the theorem
itself is achieved in two steps. First, the \PM generates a
\emph{structure lemma} (and its Isabelle/\tlaplus proof) that states
simply that the collection of leaf obligations implies the theorem.
Then, the \PM generates a proof of the theorem using the
already-certified obligations and structure lemma.  If Isabelle
accepts that proof, we are assured that the translated version of the
theorem is true in Isabelle/\tlaplus, regardless of any errors made by
the \PM.

Of course, we expect the \PM to be correct.  We now explain why it
should be by describing how it generates the leaf obligations from the
proof of a theorem. (Remember that we are considering only \tlatwo
formulas with no temporal operators.)
Formally, a theorem in \tlatwo represents a closed proof obligation in
the \tlatwo meta-logic of the form "\obl{\G ||- e}", where "\G" is a
\emph{context} containing all the declarations, definitions, facts
(previous assumptions or theorems) and the assumptions introduced in
the theorem using an \ASSUME clause (if present), and "e" is a \tlatwo
formula that is the \emph{goal} of the theorem.

A closed obligation "\obl{\G ||- e}" is \textit{true} if $e$ is
entailed by $\G$ in the formal semantics of
\tlaplus~\cite{lamport03tla}. It is said to be \emph{provable} if we
have a proof of $e$ from $\G$ in Isabelle/\tlaplus. Because we assume
Isabelle/\tlaplus to be sound, we consider any provable obligation to
be true.  A \emph{claim} is a sentence of the form "\pi:\obl{\G ||-
  e}", where $\pi$ is a \tlatwo proof.  This claim represents the
verification task that $\pi$ is a proof of the proof obligation
"\obl{\G ||- e}". The \PM generates the leaf obligations of a claim by
recursively traversing its proof, using its structure to refine the
obligation of the claim. For a non-leaf proof, each proof step
modifies the context or the goal of its obligation to produce an
obligation for its following step, and the final \QED step proves the
final form of the obligation. More precisely, every step defines a
\textit{transformation}, written "\sigma.\,\tau: \obl{\G ||- e} -->
\obl{\D ||- f}", which states that the \emph{input} obligation
"\obl{\G ||- e}" is \emph{refined} to the obligation "\obl{\D ||- f}"
by the step "\sigma.\,\tau". A step is said to be \textit{meaningful}
if the input obligation matches the form of the step. (An example of a
meaningless claim is one that involves a \TAKE\ step whose input
obligation does not have a universally quantified goal.) A claim is
meaningful if every step in it is meaningful.

The recursive generation of leaf obligations for meaningful claims and
transformations is specified using inference rules, with the
interpretation that the leaf obligations generated for the claim or
transformation at the conclusion of a rule is the union of those
generated by the claims and transformations in the premises of the
rule.  For example, the following rule is applied to generate the leaf
obligations for a claim "\pi:\obl{\G ||- e}" when $\pi$ is a sequence
of $n$ steps, for $n>1$.
\begin{gather*}
  \I{"\sigma_1.\,\tau_1 \ \ \sigma_2.\,\tau_2 \ \ \dotsb \ \ \sigma_n.\,\tau_n : \obl{\G ||- e}"}
    {"\sigma_1.\,\tau_1 : \obl{\G ||- e} --> \obl{\D ||- f}"
     &
     "\sigma_2.\,\tau_2 \ \ \dotsb \ \ \sigma_n.\,\tau_n : \obl{\D ||- f}"}
\end{gather*}
The leaf obligations of the claim in the conclusion are the union of
those of the claim and transformation in the premises.  As an example
of leaf obligations generated by a transformation, here is a rule for
the step "\sigma.\,\tau" where "\sigma" is the begin-step level token
$\s{n}$ and "\tau" is the proposition "p" with proof $\pi$.
\begin{gather*}
  \I{"\s{n}.\ p\ \PROOF\ \pi : \obl{\G ||- e} --> \obl{\G, p ||- e}"}
    {"\pi : \obl{\G, \hide{\lnot e} ||- p}"}
\end{gather*}
The rule concludes that the refinement in this step is to add "p" to
the context of the obligation, assuming that the sub-proof "\pi" is
able to establish it. The leaf obligations generated by this
transformation are the same as those of the claim in the premise of
the rule.  The goal "e" is negated and added to the context as a
hidden fact (the square brackets indicate hiding).  We can use "\lnot
e" in a \BY proof or \USE statement, and doing so can simplify
subproofs.  (Because we are using classical logic, it is sound to add
$\lnot e$ to the known facts in this way.)
The full set of such rules for every construct in the \tlatwo proof
language is given in appendix~\ref{apx}.

A claim is said to be \emph{complete} it its proof contains no omitted
subproofs. Starting from a complete meaningful claim, the \PM first
generates its leaf obligations and \textit{filters} the hidden
assumptions from their contexts. (Filtration amounts to deleting
hidden facts and replacing hidden operator definitions with
declarations.) The \PM then asks the back-end provers to find proofs
of the filtered obligations, which are used to certify the obligations
in Isabelle/\tlaplus.
The \PM next writes an Isar proof of the obligation of the complete
meaningful claim that uses its certified filtered leaf
obligations. The following meta-theorem (proved in
Appendix~\ref{apx:constraints}) ensures that the \PM can do this for
all complete meaningful claims.
\begin{thm}[Structural Soundness Theorem] \label{thm:meaning}
  If "\pi:\obl{\G ||- e}" is a complete meaningful claim and every
  leaf obligation it generates is provable after filtering hidden
  assumptions , then "\obl{\G ||- e}" is provable.
\end{thm}

\noindent%
Isabelle/\tlaplus then uses this proof to certify the obligation of
the claim. From the assumptions that the Isabelle/\tlaplus
axiomatization is faithful to the semantics of \tlatwo and that the
embedding of \tlatwo into Isabelle/\tlaplus is sound, it follows that
the obligation is true.

\section{Verifying Proof Obligations}
\label{sec:backend}

Once the \PM generates the leaf obligations, it must send them to the
back-end provers.  The one non-obvious part of doing this is deciding
whether definitions should be expanded by the \PM or by the prover.
This is discussed in Section~\ref{sec:backend.pm}.  We then describe
the state of our two current back-end provers, Isabelle/\tlaplus and
Zenon.

\subsection{Expanding Definitions}
\label{sec:backend.pm}

Expansion of usable definitions cannot be left entirely to the
back-end prover.  The \PM itself must do it for two reasons:
\begin{icom}
\item It must check that the current goal has the right form for a
  \TAKE, \WITNESS, or \HAVE\ step to be meaningful, and this can require
  expanding definitions.

\item The encoding of \tlaplus in the back-end prover's logic would be
  unsound if a modal operator like prime~($'$) were encoded as a
  non-modal operator. Hence, encoding a definition like $O(x)\DEF x'$
  as an ordinary definition in the prover's logic would be unsound.
  All instances of such operators must be removed by expanding their
  definitions before a leaf obligation is sent to the back-end prover.
  Such operator definitions seldom occur in actual \tlaplus
  specifications, but the \PM must be able to deal with them.
\end{icom}
Another reason for the \PM to handle definition expansion is that the
Isabelle/\tlaplus object logic does not provide a direct encoding of
definitions made within proofs. We plan to reduce the amount of
trusted code in the \PM by lambda-lifting all usable definitions out
of each leaf obligation and introducing explicit operator definitions
using Isabelle's meta equality ($\equiv$). These definitions will be
expanded before interacting with Isabelle.

\subsection{Isabelle/\tlaplus}
\label{sec:backend.isa}

The core of \tlatwo is being encoded as a new object logic
Isabelle/\tlaplus in the proof assistant
Isabelle~\cite{paulson:isabelle}.  One of Isabelle's distinctive
features that similar proof assistants such as Coq~\cite{coq} or
HOL~\cite{gordon:hol,harrison:hol} lack is genericity with respect to
different logics. The base system Isabelle/Pure provides the trusted
kernel and a framework in which the syntax and proof rules of object
logics can be defined. We have chosen to encode \tlatwo as a separate
object logic rather than add it on top of one of the existing logics
(such as ZF or HOL). This simplifies the translation and makes it
easier to interpret the error messages when Isabelle fails to prove
obligations. A strongly typed logic such as HOL would have been
unsuitable for representing \tlatwo, which is untyped. Isabelle/ZF
might seem like a natural choice, but differences between the way it
and \tlaplus define functions and tuples would have made the
encoding awkward and would have prevented us from reusing existing
theories. Fortunately, the genericity of Isabelle helped us not only
to define the new logic, but also to instantiate the main automated
proof methods, including rewriting, resolution- and tableau provers,
and case-based and inductive reasoning.  Adding support for more
specialized reasoning tools such as proof-producing SAT
solvers~\cite{fontaine:automation} or SMT solvers such as
haRVey~\cite{deharbe:decision} will be similarly helped by existing
generic interfaces.

The current encoding supports only a core subset of \tlatwo, including
propositional and first-order logic, elementary set theory, functions,
and the construction of natural numbers.  Support for arithmetic,
strings, tuples, sequences, and records is now being added; support
for the modal part of \tlatwo{} (variables, priming, and temporal
logic) will be added later.  Nevertheless, the existing fragment can
already be used to test the interaction of the \PM with Isabelle and
other back-end provers.  As explained above, Isabelle/\tlaplus is used
both as a back-end prover and to check proof scripts produced by other
back-end provers such as Zenon.  If it turns out to be necessary, we
will enable the user to invoke one of Isabelle's automated proof
methods (such as \texttt{auto} or \texttt{blast}) by using a dummy
theorem, as explained at the end of
Section~\ref{sec:proof-language.lang}.  If the method succeeds, one
again obtains an Isabelle theorem.  Of course,
Isabelle/\tlaplus can also be used independently of the \PM, which is
helpful when debugging tactics.

\subsection{Zenon}
\label{sec:backend.zenon}

Zenon~\cite{bonichon07lpar} is a tableau prover for classical
first-order logic with equality that was initially designed to output
formal proofs checkable by Coq~\cite{coq}.
Zenon outputs proofs in an automatically-checkable format and it is
easily extensible with new inference rules.  One of its design goals
is predictability in solving simple problems, rather than high
performance in solving some hard problems.  These characteristics make
it well-suited to our needs.

We have extended Zenon to output Isar proof scripts for
Isabelle/\tlaplus theorems, and the \PM uses Zenon as a back-end
prover, shipping the proofs it produces to Isabelle to certify the
obligation.  We have also extended Zenon with direct support for the
\tlatwo logic, including definitions and rules about sets and
functions.  Adding support in the form of rules (instead of axioms) is
necessary because some rules are not expressible as first-order
axioms, notably the rules about the set constructs:
\begin{gather*} \small
  \I[\textit{subsetOf}]{e\in\{x\in S:P\}}{e\in S & P[x:=e]}
  \qquad
  \I[\textit{setOfAll}]{e\in\{d:x\in S\}}{\exists y\in S:e=d[x:=y]}
\end{gather*}
Even for the rules that are expressible as first-order axioms, adding
them as rules makes the proof search procedure much more efficient in
practice.  The most important example is extensionality: when set
extensionality and function extensionality are added as axioms, they
apply to every equality deduced by the system, and pollute the search
space with large numbers of irrelevant formulas.  By adding them as
rules instead, we can use heuristics to apply them only in cases where
they have some chance of being useful.

Adding support for arithmetic, strings, tuples, sequences, and records
will be done in parallel with the corresponding work on
Isabelle/\tlaplus, to ensure that Zenon will produce proof scripts
that Isabelle/\tlaplus will be able to check. Temporal logic will be
added later.
We also plan to interface Zenon with Isabelle, so it can be called by
a special Isabelle tactic the same way other tools are.  This will
simplify the \PM by giving it a uniform interface to the back-end
provers.  It will also allow using Zenon as an Isabelle tactic
independently of \tlaplus.

\section{Conclusions and Future Work}
\label{sec:conclusions}

We have presented a hierarchically structured proof language for
\tlaplus.  It has several important features that help in managing the
complexity of proofs.  The hierarchical structure means that changes
made at any level of a proof are contained inside that level, which
helps construct and maintain proofs.  Leaf proofs can be omitted and
the resulting incomplete proof can be checked.  This allows different parts
of the proof to be written separately, in a non-linear fashion.
The more traditional linear proof style, in which steps that have not
yet been proved can be used only if explicitly added as hypotheses,
encourages proofs that use many separate lemmas.  Such proofs lack the
coherent structure of a single hierarchical proof.

The proof language lets the user freely and repeatedly make facts and
definitions usable or hidden.  Explicitly stating what is being used
to prove each step makes the proof easier for a human to understand.
It also aids a back-end prover by limiting its search for a proof
to ones that use only necessary facts.

There are other declarative proof languages that are similar to
\tlatwo. Isar~\cite{isar} is one such language, but it has significant
differences that encourage a different style of proof development.
For example, 
it provides an \emph{accumulator} facility to avoid explicit
references to proof steps.  This is fine for short proofs, but in
our experience does not work well for long proofs that are typical of
algorithm verification that \tlatwo{} targets.
Moreover, because Isabelle is designed for interactive use, the
effects of the Isar proof commands are not always easily predictable, and
this encourages a linear rather than hierarchical proof development
style.
\nocite{rudnicki:mizar}%
The Focal Proof Language~\cite{focal} is essentially a subset of the
\tlatwo proof language.  Our experience with hierarchical proofs in
Focal provides additional confidence in the attractiveness of our
approach.  We know of no declarative proof language that has as
flexible a method of using and hiding facts and definitions as that of
\tlatwo.

The \PM transforms a proof into a collection of proof obligations to
be verified by a back-end prover.  Its current version handles proofs
of theorems in the non-temporal fragment of \tlaplus that do not
involve module instantiation (importing of modules with substitution).
Even with this limitation, the system can be useful for many
engineering applications.  We are therefore concentrating on making
the \PM and its back-end provers handle this fragment of \tlaplus
effectively before extending them to the complete language.  The major
work that remains to be done on this is to complete the Zenon and
Isabelle inference rules for reasoning about the built-in constant
operators of \tlaplus.  There are also a few non-temporal aspects of
the \tlatwo language that the \PM does not yet handle, such as
subexpression naming.  We also expect to extend the \PM to support
additional back-end provers, including decision procedures for
arithmetic and for propositional temporal logic.

We do not anticipate that any major changes will be needed to the
\tlatwo proof language.  We do expect some minor tuning as we get more
experience using it. For example, we are not sure whether local
definitions should be usable by default.  A graphical user interface
is being planned for the \tlaplus tools, including the \PM. It will
support the non-linear development of proofs that the language and the
proof system allow.

\bibliographystyle{plain}
\bibliography{submission}

\appendix

\section{Details of the \PM}
\label{apx}

We shall now give a somewhat more formal specification of the \PM and
prove the key Structural Soundness Theorem~\ref{thm:meaning}.  We
begin with a quick summary of the abstract syntax of \tlatwo proofs,
ignoring the stylistic aspects of their concrete
representation. (See~\cite{lamport08tla+2} for a more detailed
presentation of the proof language.)

\begin{defn}[\tlatwo Proof Language] \label{defn:proof-language}
  \tlatwo \emph{proofs}, \emph{non-leaf proofs}, \emph{proof steps}
  and \emph{begin-step} tokens have the following syntax, where "n"
  ranges over natural numbers, "l" over labels, "e" over expressions,
  "\Phi" over lists of expressions, "o" over operator definitions,
  "\P" over sets of operator names, "\vec \beta" over lists of binders
  (\ie, constructs of the form "x" and "x \in e" used to build
  quantified expressions), and "\alpha" over expressions or \ASSUME
  \ldots \PROVE forms.
  \begin{quote} \itshape
    \begin{tabbing}
      (Proofs) \hspace{4.5em}  \= "\pi" \LSP \= "::=" \ \= "\OBVIOUS OR \OMITTED OR \BY\ \Phi\ \DEFS\ \P OR \Pi" \\
      (Non-leaf proofs) \> "\Pi" \> "::=" \> "\sigma.\ \QED\ \PROOF\ \pi" \\
                   \>             \> "\ \ \ |"  \> "\sigma.\ \tau\quad \Pi" \\
      (Proof steps) \> "\tau" \> "::=" \> "\USE\ \Phi\ \DEFS\ \P OR \HIDE\ \Phi\ \DEFS\ \P OR \DEFINE\ o" \\
                    \>          \> "\ \ \ |" \> "\HAVE\ e OR \TAKE\ \vec \beta OR \WITNESS\ \Phi" \\
                    \>          \> "\ \ \ |" \> "\alpha\ \PROOF\ \pi OR \SUFFICES\ \alpha\ \PROOF\ \pi 
                                                    OR \PICK\ \vec \beta : e\ \PROOF\ \pi" \\
      (Begin-step tokens) \> "\sigma" \> "::=" \> "\s n OR \s n l"
    \end{tabbing}
  \end{quote}
  A proof that is not a non-leaf proof is called a \emph{leaf
    proof}. The level numbers of a non-leaf proof must all be the
  same, and those in the subproof of a step (that is, the "\pi" in
  "\alpha\ \PROOF\ \pi", \etc.) must be strictly greater than that of
  the step itself.
\end{defn}

\subsection{The Meta-Language}
\label{apx:context}

The \PM uses proofs in the \tlatwo proof language
(Definition~\ref{defn:proof-language}) to manipulate constructs in the
meta-language of \tlatwo. This meta-language naturally has no
representation in \tlatwo itself; we define its syntax formally as
follows.

\begin{defn}[Meta-Language] \label{defn:syntax}
  The \tlatwo meta-language consists of \emph{obligations},
  \emph{assumptions} and \emph{definables} with the following syntax,
  where "e" ranges over \tlatwo expressions, "x" and "o" over \tlatwo
  identifiers, and "\vec x" over lists of \tlatwo identifiers.
  \begin{quote}
    \begin{tabbing}
      (Obligations) \SP \= "\phi" \LSP \= "::=" \ \= "\obl{h_1, \dotsc, h_n ||- e}" \` ("n \ge 0") \\
      (Assumptions) \> "h" \> "::=" \> "\NEW x OR o \DEF \delta OR \phi OR \hide{o \DEF \delta} OR \hide{\phi}" \\
      (Definables) \> "\delta" \> "::=" \> "\phi OR \LAMBDA\ \vec x : e"
    \end{tabbing}
  \end{quote}
  The expression after "||-" in an obligation is called its
  \emph{goal}. An assumption written inside square brackets "\hide{\
  }" is said to be \emph{hidden}; otherwise it is \emph{usable}. For
  any assumption "h", we write "\unhide h" (read: "h" \emph{made
    usable}) to stand for "h" with its brackets removed if it is a
  hidden assumption, and to stand for "h" if it is not hidden.
  A list of assumptions is called a \emph{context}, with the empty
  context written as "\nil"; we let "\G", "\D" and "\W" range over
  contexts, with "\G, \D" standing for the context that is the
  concatenation of "\G" and "\D". The context "\unhide{\G}" is "\G"
  with all its hidden assumptions made usable.
  The obligation "\obl{\nil ||- e}" is written simply as "e". The
  assumptions "\NEW x", "o \DEF \delta" and "\hide{o \DEF \delta}"
  \emph{bind} the identifiers "x" and "o" respectively. We write "x
  \in \G" if "x" is bound in "\G" and "x \notin \G" if "x" is not
  bound in "\G". The context "\G, h" is considered syntactically
  well-formed iff "h" does not bind an identifier already bound in
  "\G".
\end{defn}

\noindent%
An obligation is a statement that its goal follows from the
assumptions in its context. \tlatwo already defines such a statement
using \ASSUME \ldots \PROVE, but the contexts in such statements have
no hidden assumptions or definitions. 
(To simplify the presentation,
we give the semantics of a slightly enhanced proof language
where proof steps are allowed to mention obligations instead of just
\tlatwo \ASSUME \ldots \PROVE statements.) We define an embedding of
obligations into Isabelle/\tlaplus propositions, which we take as the
ultimate primitives of the \tlatwo meta-logic.

\begin{defn} \label{defn:isabelle-embedding}
  The Isabelle/\tlaplus \emph{embedding} "\isa{-}" of obligations,
  contexts and definables is as follows:
  \begin{align*}
    \begin{aligned}
    \isa{"\G ||- e"} &\ =\  "\isa{\unhide{\G}}\, e" \\[1ex]
    \isa{"\LAMBDA\ \vec x : e"} &\ =\ "\lambda \vec x.\ e" 
    \end{aligned}
    \qquad \qquad
    \begin{aligned}
    \isa{\,\nil\,} &\ =\  \\
    \isa{"\G, \NEW x"} &\ =\  \isa{\G}\,\And x. \\
    \isa{"\G, o \DEF \delta"} &\ =\  \isa{\G}\,\And o.\,\bigl(o \equiv \isa{\delta}\bigr)\ "==>" \\
    \isa{"\G, \phi"} &\ =\  \isa{\G}\,\bigl(\isa{\phi}\bigr)\ "==>" \\[1ex]
    \end{aligned}
  \end{align*}
\end{defn}

\noindent%
For example, "\isa{\NEW P, \hide{\obl{\NEW x ||- P(x)}} ||- \forall x
  : P(x)} = \And P.\ \left(\And x.\ P(x)\right) ==> \forall x : P(x)".
Note that usable and hidden assumptions are treated identically for
the provability of an obligation.

The embedding of ordinary \tlatwo expressions is the identity because
Isabelle/\tlaplus contains \tlatwo expressions as part of its object
syntax. Thus, we do not have to trust the embedding of ordinary
\tlatwo expressions, just that of the obligation language.  In
practice, some aspects of \tlatwo expressions, such as the
indentation-sensitive conjunction and disjunction lists, are sent by
the \PM to Isabelle using an indentation-insensitive
encoding.%
While Isabelle/\tlaplus can implicitly generalize over the free
identifiers in a lemma, we shall be explicit about binding and
consider obligations provable only if they are closed.

\begin{defn}[Well-Formed Obligations] \label{defn:wfo}
  The obligation "\obl{\G ||- e}" is said to be \emph{well-formed} iff
  it is closed and "\isa{\G ||- e}" is a well-typed proposition of
  Isabelle/\tlaplus.
\end{defn}

\begin{defn}[Provability] \label{defn:provable}
  The obligation "\obl{\G ||- e}" is said to be \emph{provable} iff it
  is well-formed and "\isa{\G ||- e}" is certified by the Isabelle
  kernel to follow from the axioms of the Isabelle/\tlaplus object
  logic.
\end{defn}

\noindent%
We trust Isabelle/\tlaplus to be sound with respect to the semantics
of \tlatwo, and therefore provability to imply truth. Formally, we
work under the following \textit{trust} axiom.

\begin{axm}[Trust] \label{axm:trust} \mbox{}
  If "\phi" is provable, then it is true.
\end{axm}

\noindent%
We state a number of useful facts about
obligations (which are all theorems in Isabelle/\tlaplus),
omitting their trivial proofs. The last one
(Fact~\ref{thm:classic}) is true because \tlaplus is based on
classical logic.

\begin{fac}[Definition] \label{thm:definition}
  If "\obl{\G, \NEW o, \D ||- e}" is provable, then "\obl{\G, o \DEF
    \delta, \D ||- e}" is provable if it is well-formed.
\end{fac}

\begin{fac}[Weakening] \label{thm:weaken}
  If "\obl{\G, \D ||- e}" is provable, then "\obl{\G, h, \D ||- e}" is
  provable if it is well-formed.
\end{fac}

\begin{fac}[Expansion] \label{thm:expand}
  If "\obl{\G, o \DEF \delta, \D ||- e}" is provable, then "\obl{\G, o
    \DEF \delta, \D[o := \delta] ||- e[o := \delta]}" is provable.
\end{fac}

\begin{fac}[Strengthening] \label{thm:delete}
  If "\obl{\G, \NEW o, \D ||- e}" or "\obl{\G, o \DEF \delta, \D ||-
    e}" is provable and "o" is not free in "\obl{\D ||- e}", then
  "\obl{\G, \D ||- e}" is provable.
\end{fac}

\begin{fac}[Cut] \label{thm:cut}
  If "\obl{\G, \D ||- e}" is provable and "\obl{\G, \obl{\D ||- e}, \W ||- f}" is provable,
  then "\obl{\G, \W ||- f}" is provable.
\end{fac}

\begin{fac} \label{thm:classic}
  If "\obl{\G, \lnot e, \D ||- e}" is provable, then "\obl{\G, \D ||- e}"
  is provable.
\end{fac}

\noindent%
The \USE/\HIDE \DEFS steps change the visibility of definitions in a
context (Definition~\ref{defn:use/hide} below). Note that changing the
visibility of a definition does not affect the provability of an obligation
because the Isabelle embedding (Definition~\ref{defn:isabelle-embedding})
makes all hidden definitions usable.

\pagebreak[2]

\begin{defn} \label{defn:use/hide}
  If "\G" is a context and "\P" a set of operator names, then:
  \begin{ecom}
  \item \emph{"\G" with "\P" made usable}, written "\G \USING \P", is
    constructed from "\G" by replacing all assumptions of the form
    "\hide{o \DEF \delta}" in "\G" with "o \DEF \delta" for every "o
    \in \P".
  \item \emph{"\G" with "\P" made hidden}, written "\G \HIDING \P", is
    constructed from "\G" by replacing all assumptions of the form "o
    \DEF \delta" in "\G" with "\hide{o \DEF \delta}" for every "o \in
    \P".
  \end{ecom}
\end{defn}

\def\refl#1{\bigl\|\,#1\,\bigr\|}

\noindent%
A sequence of binders "\vec \beta" in the \tlatwo expressions "\forall
\vec \beta : e" or "\exists \vec \beta : e" can be reflected as
assumptions.

\begin{defn}[Binding Reflection] \label{defn:binding-reflection}
  If "\vec \beta" is a list of binders with each element of the form
  "x" or "x \in e", then the \emph{reflection} of "\vec \beta" as
  assumptions, written "\refl{\vec \beta}", is given inductively as
  follows.
  \begin{align*}
    \refl{\,\nil\,} &= \;\nil &
    \refl{\vec \beta, x} &= \refl{\vec \beta}, \NEW x &
    \refl{\vec \beta, x \in e} &= \refl{\vec \beta}, \NEW x, x \in e
  \end{align*}
\end{defn}

\subsection{Interpreting Proofs}
\label{apx:proof-transformation}

Let us recall some definitions from section~\ref{sec:obligations}.

\begin{defn}[Claims and Transformations] \label{defn:check/trans-def}
  A \emph{claim} is a judgement of the form "\pi : \obl{\G ||- e}"
  where "\pi" is a \tlatwo proof. A \emph{transformation} is a
  judgement of the form "\sigma.\,\tau : \obl{\G ||- e} --> \obl{\D
    ||- f}" where "\sigma" is a begin-step token and "\tau" a proof
  step. A claim (respectively, transformation) is said to be
  \emph{complete} if its proof (respectively, proof step) does not
  contain any occurrence of the leaf proof \OMITTED.
\end{defn}

\noindent%
The \PM generates leaf obligations for a claim using two mutually
recursive procedures, \textit{checking} and \textit{transformation},
specified below using the formalism of a \textit{primitive
  derivation}.

\begin{defn} \label{defn:primitive}
  A \emph{primitive derivation} is a derivation constructed using
  inferences of the form
  \begin{gather*} \small
    \Ic{E}{\DD_1 & \dotsb & \DD_n} \tag*{"(n \ge 0)"}
  \end{gather*}
  where "E" is either a claim or a transformation, and "\DD_1, ...,
  \DD_n" are primitive derivations or obligations. An obligation at
  the leaf of a primitive derivation is called a \emph{leaf
    obligation}.
\end{defn}

\begin{defn}[Checking and Transformation] \label{defn:check/trans-proc}
  The primitive derivations of a claim or transformation are
  constructed using the following \emph{checking} and
  \emph{transformation} rules.
  \begin{ecom} \small
  \item \emph{Checking} rules
    \begin{gather*}
      \I[\OBVIOUS]{"\OBVIOUS : \obl{\G ||- e}"}{"\obl{\G ||- e}"}
      \LSP
      \I[\OMITTED]{"\OMITTED : \obl{\G ||- e}"}{}
      \\[1ex]
      \I[\BY]{"\BY\ \Phi\ \DEFS\ \P : \obl{\G ||- e}"}
        {"\s{0}.\ \USE\ \Phi\ \DEFS\ \P" : "\obl{\G ||- e} --> \obl{\D ||- f}"
         &
         "\obl{\D ||- f}"}
      \\[1ex]
      \I[\QED]{"\sigma.\ \QED\ \PROOF\ \pi : \obl{\G ||- e}"}
        {"\pi : \obl{\G ||- e}"}
      \SP
      \I[non-\QED]{"\sigma.\,\tau\ \ \Pi : \obl{\G ||- e}"}
        {"\sigma.\,\tau : \obl{\G ||- e} --> \obl{\D ||- f}"
         &
         "\Pi : \obl{\D ||- f}"}
    \end{gather*}  
  \item \emph{Transformation}
    \begin{gather*} \small
      \I[\USE\ \DEFS]{"\sigma.\ \USE\ \Phi\ \DEFS\ \P : \obl{\G ||- e} --> \obl{\D ||- f}"}
        {"\sigma.\ \USE\ \Phi : \obl{\G \USING \P ||- e} --> \obl{\D ||- f}"}
      \\[1ex]
      \I[\HIDE\ \DEFS]{"\sigma.\ \HIDE\ \Phi\ \DEFS\ \P : \obl{\G ||- e} --> \obl{\D \HIDING \P ||- f}"}
         {"\sigma.\ \HIDE\ \Phi : \obl{\G ||- e} --> \obl{\D ||- f}"}
      \\
      \I[\DEFINE ("o \notin \G")]{"\sigma.\ \DEFINE\ o \DEF \delta : \obl{\G ||- e} --> \obl{\G, \hide{o \DEF \delta} ||- e}"}
      \\
      \I["\USE_0"]{"\sigma.\ \USE\ \nil : \obl{\G ||- e} --> \obl{\G ||- e}"}
      \SP
      \I["\HIDE_0"]{"\sigma.\ \HIDE\ \nil : \obl{\G ||- e} --> \obl{\G ||- e}"}
      \\[1ex]
      \I["\USE_1"]{"\sigma.\ \USE\ \Phi, \obl{\G_0 ||- e_0} : \obl{\G ||- e} --> \obl{\D, \obl{\G_0 ||- e_0} ||- f}"}
        {"\sigma.\ \USE\ \Phi : \obl{\G ||- e} --> \obl{\D ||- f}"
         &
         "\obl{\unhide{\D}, \G_0 ||- e_0}"
        }
      \\[1ex]
      \I["\HIDE_1"]{"\sigma.\ \HIDE\ \Phi, \phi : \obl{\G_0, \phi, \G_1 ||- e} --> \obl{\D ||- f}"}
        {"\sigma.\ \HIDE\ \Phi : \obl{\G_0, hide{\phi}, \G_1 ||- e} --> \obl{\D ||- f}"}
      \\
      \I["\TAKE_0"]{"\sigma.\ \TAKE\ \nil : \obl{\G ||- e} --> \obl{\G ||- e}"}
      \LSP
      \I["\WITNESS_0"]{"\sigma.\ \WITNESS\ \nil : \obl{\G ||- e} --> \obl{\G ||- e}"}
      \\[1ex]
      \I["\TAKE_1"]{"\sigma.\ \TAKE\ u, \vec\beta : \obl{\G ||- \forall x : e} --> \obl{\D ||- f}"}
        {"\sigma.\ \TAKE\ \vec\beta : \obl{\G, \NEW u ||- e [x := u]} --> \obl{\D ||- f}"}
      \\[1ex]
      \I["\TAKE_2"]{"\sigma.\ \TAKE\ u \in T, \vec\beta : \obl{\G ||- \forall x \in S : e} --> \obl{\D ||- f}"}
        {"\obl{\G ||- S \subseteq T}"
         &
         "\sigma.\ \TAKE\ \vec\beta : \obl{\G, \NEW u, u \in T ||- e [x := u]} --> \obl{\D ||- f}"}
      \\[1ex]
      \I["\WITNESS_1"]{"\sigma.\ \WITNESS\ w, \W : \obl{\G ||- \exists x : e} --> \obl{\D ||- f}"}
        {"\sigma.\ \WITNESS\ \W : \obl{\G ||- e [x := w]} --> \obl{\D ||- f}"}
      \\[1ex]
      \I["\WITNESS_2"]{"\sigma.\ \WITNESS\ w \in T, \W : \obl{\G ||- \exists x \in S : e} --> \obl{\D ||- f}"}
        {"\obl{\G ||- T \subseteq S}"
         &
         "\obl{\G ||- w \in T}"
         &
         "\sigma.\ \WITNESS\ \W : \obl{\G, w \in T ||- e [x := w]} --> \obl{\D ||- f}"}
      \\[1ex]
      \I[\HAVE]{"\sigma.\ \HAVE\ g : \obl{\G ||- e => f} --> \obl{\G, g ||- f}"}
        {"\obl{\G, e ||- g}"}
      \\[1ex]
      \I[$\ASSERT_1$]{"\s n.\ \obl{\D ||- f}\ \PROOF\ \pi : \obl{\G ||- e} --> \obl{\G, \obl{\D ||- f} ||- e}"}
        {"\pi : \obl{\G, \hide{\lnot e}, \D ||- f}"}
      \\[1ex]
      \I[$\ASSERT_2$]{"\s n l.\ \obl{\D ||- f}\ \PROOF\ \pi : \obl{\G ||- e} --> \obl{\G, \s n l \DEF \obl{\D ||- f}, \hide{\s n l} ||- e}"}
        {"\pi : \obl{\G, \s n l \DEF \obl{\D ||- f}, \hide{\lnot e}, \D ||- f}"}
      \\[1ex]
      \I[\CASE]{"\sigma.\ \CASE\ g\ \PROOF\ \pi : \obl{\G ||- e}
                       --> \obl{\D ||- f}"}
        {"\sigma.\ \obl{g ||- e}\ \PROOF\ \pi : \obl{\G ||- e} --> \obl{\D ||- f}"}
      \\[1ex]
      \I[$\SUFFICES_1$]{"\s n.\ \SUFFICES\ \obl{\D ||- f}\ \PROOF\ \pi : \obl{\G ||- e} --> \obl{\G, \hide{\lnot e}, \D ||- f}"}
        {"\pi : \obl{\G, \obl{\D ||- f} ||- e}"}
      \\[1ex]
      \I[$\SUFFICES_2$]{"\s n l.\ \SUFFICES\ \obl{\D ||- f}\ \PROOF\ \pi : \obl{\G ||- e} --> \obl{\G, \s n l \DEF \obl{\D ||- f}, \hide{\lnot e}, \D ||- f}"}
        {"\pi : \obl{\G, \s n l \DEF \obl{\D ||- f}, \hide{\s n l} ||- e}"}
      \\[1ex]
      \I[\PICK]{"\sigma.\ \PICK\ \vec\beta : p\ \PROOF\ \pi :
                  \obl{\G ||- e} --> \obl{\G, \refl{\vec \beta}, p ||- e}"}
        {"\pi : \obl{\G ||- \exists \vec \beta : p}"}
    \end{gather*}
  \end{ecom}
\end{defn}

\noindent%
The inference rules in the above definition are deterministic: the
conclusion of each rule uniquely determines the premises. However, the
rules are partial; for example, there is no rule that concludes a
transformation of the form "\sigma.\,\TAKE\ x \in S : \obl{\G ||- B
  \land C} --> \obl{\D ||- f}".

\begin{defn} \label{defn:meaningful}
  A claim or a transformation is said to be \emph{meaningful} if it
  has a primitive derivation.
\end{defn}

\begin{defn}[Generating Leaf Obligations]
  A meaningful claim or transformation is said to \emph{generate} the
  leaf obligations of its primitive derivation.
\end{defn}

\noindent%
In the rest of this appendix we limit our attention to complete
meaningful claims and transformations.

\subsection{Correctness}
\label{apx:correctness}

If the leaf obligations generated by a complete meaningful claim are
provable, then the obligation in the claim itself ought to be
provable. In this section we prove this theorem by analysis of the
checking and transformation rules.

\begin{defn}[Provability of Claims and Transformation]
  \label{defn:proc-provable} \mbox{}
  \begin{ecom}
  \item The claim "\pi : \obl{\G ||- e}" is \emph{provable} iff it is
    complete and meaningful and the leaf obligations it generates are
    all provable.
  \item The transformation "\sigma.\,\tau : \obl{\G ||- e} --> \obl{\D
      ||- f}" is \emph{provable} iff it is complete and meaningful and
    the leaf obligations it generates are all provable.
  \end{ecom}
\end{defn}

\begin{thm}[Correctness] \label{thm:correctness} \mbox{}
  \begin{ecom}[\LSP (1)]
  \item If "\pi : \obl{\G ||- e}" is provable, then "\obl{\G ||- e}"
    is provable.
  \item If "\sigma.\,\tau : \obl{\G ||- e} --> \obl{\D ||- f}" is
    provable and "\obl{\D ||- f}" is provable, then "\obl{\G ||-
      e}" is provable.
  \end{ecom}
\end{thm}

\begin{proof} \small
  Let "\DD" be the primitive derivation for the claim in (1) and let "\EE"
  be the primitive derivation for the transformation in (2). The proof
  will be by lexicographic induction on the structures of "\DD" and
  "\EE", with a provable transformation allowed to justify a
  provable claim.

  \begin{ecom}[{$\s1$}1.]
  \item If "\pi : \obl{\G ||- e}" is provable, then "\obl{\G ||-
      e}" is provable.
    \begin{ecom}[{$\s2$}1.] \setlength{\itemsep}{6pt}
    \item \Case "\pi" is "\OBVIOUS", \ie,
      $\DD = \Im[\OBVIOUS.]{"\OBVIOUS : \obl{\G ||- e}"}{"\obl{\G ||- e}"}$
      \Trivial

    \item \Case "\pi" is "\OMITTED" is impossible because "\pi :
      \obl{\G ||- e}" is complete.

    \item \Case "\pi" is "\BY\ \Phi\ \DEFS\ \P", \ie,
      \begin{gather*}
        \DD = 
        \Im[\BY.]{"\BY\ \Phi\ \DEFS\ \P : \obl{\G ||- e}"}
           {\deduce{"\s{0}.\ \USE\ \Phi\ \DEFS\ \P" : "\obl{\G ||- e} --> \obl{\D ||- f}"}{\EE_0}
            &
            "\obl{\D ||- f}"}
      \end{gather*}
      \begin{ecom}[{$\s3$}1.]
      \item "\obl{\D ||- f}" is provable
        \by Definition~\ref{defn:proc-provable}.
      \item \Qed
        \by \s31, i.h. (inductive hypothesis) for "\EE_0".
      \end{ecom}

    \item \Case "\pi" is "\sigma.\, \QED\ \PROOF\ \pi_0", \ie, 
      $
      \DD =
      \Im[\QED.]{"\sigma.\ \QED\ \PROOF\ \pi_0 : \obl{\G ||- e}"}
         {\deduce{"\pi_0 : \obl{\G ||- e}"}{\DD_0}}$
      \by i.h. for "\DD_0".

    \item \Case "\pi" is "\sigma.\,\tau\ \ \Pi", \ie,
      \begin{gather*}
        \DD =
        \Im[non-\QED.]{"\sigma.\,\tau\ \ \Pi : \obl{\G ||- e}"}
          {\deduce{"\sigma.\,\tau : \obl{\G ||- e} --> \obl{\D ||- f}"}{\EE_0}
           &
           \deduce{"\Pi : \obl{\D ||- f}"}{\DD_0}
          }
      \end{gather*}

      \begin{ecom}[{$\s3$}1.]
      \item [\s31.] "\obl{\D ||- f}" is provable
        \by i.h. for "\DD_0".
      \item [\s33.] \Qed
        \by \s31, i.h. for "\EE_0".
      \end{ecom}
    \item \Qed \by \s21, \ldots, \s25.
    \end{ecom}

  \item If "\sigma.\,\tau : \obl{\G ||- e} --> \obl{\D ||- f}" is
    provable and "\obl{\D ||- f}" is provable, then "\obl{\G ||-
      e}" is provable.
    \begin{ecom}[{$\s2$}1.] \setlength{\itemsep}{6pt}

    \item \Case "\tau" is "\USE\ \Phi\ \DEFS\ \P", \ie,
      \begin{gather*}
        \EE =
        \Im[\USE\ \DEFS.]{"\sigma.\ \USE\ \Phi\ \DEFS\ \P : \obl{\G ||- e} --> \obl{\D ||- f}"}
          {\deduce{"\sigma.\ \USE\ \Phi : \obl{\G \USING \P ||- e} --> \obl{\D ||- f}"}{\EE_0}}
      \end{gather*}

      \begin{ecom}[{$\s3$}1.]
      \item "\obl{\G \USING \P ||- e}" is provable
        \by i.h. for "\EE_0".
      \item \Qed
        \by \s31, Definition~\ref{defn:use/hide}.
      \end{ecom}

    \item \Case "\tau" is "\HIDE\ \Phi\ \DEFS\ \P", \ie,
      \begin{gather*}
        \EE =
        \Im[\HIDE\ \DEFS.]{"\sigma.\ \HIDE\ \Phi\ \DEFS\ \P : \obl{\G ||- e} --> \obl{\D \HIDING \P ||- f}"}
           {\deduce{"\sigma.\ \HIDE\ \Phi : \obl{\G ||- e} --> \obl{\D ||- f}"}{\EE_0}}
      \end{gather*}

      \begin{ecom}[{$\s3$}1.]
      \item "\obl{\D ||- f}" is provable
        \by provability of "\obl{\D \HIDING \P ||- f}" and Definition~\ref{defn:use/hide}.
      \item \Qed
        \by \s31, i.h. for "\EE_0".
      \end{ecom}

    \item \Case "\tau" is "\DEFINE\ o \DEF \delta" with "o \notin \G", \ie, 
      \begin{gather*}
        \EE =
        \Im[\DEFINE.]{"\sigma.\ \DEFINE\ o \DEF \delta : \obl{\G ||- e} --> \obl{\G, \hide{o \DEF \delta} ||- e}"}{}
      \end{gather*}

      \begin{ecom}[{$\s3$}1.]
      \item "o" is not free in "e"
        \by "o \notin \G" and closedness of "\obl{\G ||- e}".
      \item \Qed
        \by \s31, strengthening (Fact~\ref{thm:delete}).
      \end{ecom}

    \item \Case "\tau" is "\USE\ \nil", \ie,
      $
      \EE = 
      \Im["\USE_0".]{"\sigma.\ \USE\ \nil : \obl{\G ||- e} --> \obl{\G ||- e}"}{}
      $
      \Trivial

    \item \Case "\tau" is "\HIDE\ \nil", \ie,
      $
      \EE = 
      \Im["\HIDE_0".]{"\sigma.\ \HIDE\ \nil : \obl{\G ||- e} --> \obl{\G ||- e}"}{}
      $
      \Trivial

    \item \Case "\tau" is "\USE\ \Phi, \phi", \ie,
      \begin{gather*}
        \EE = 
        \Im["\USE_1"]{"\sigma.\ \USE\ \Phi, \obl{\G_0 ||- e_0} : \obl{\G ||- e} --> \obl{\D_0, \obl{\G_0 ||- e_0} ||- f}"}
          {\deduce{"\sigma.\ \USE\ \Phi : \obl{\G ||- e} --> \obl{\D_0 ||- f}"}{\EE_0}
           &
           "\obl{\unhide{\D_0}, \G_0 ||- e_0}"
          }      
      \end{gather*}

      \begin{ecom}[{$\s3$}1.]
      \item "\obl{\unhide{\D_0}, \G_0 ||- e_0}" is provable
        \by Definition~\ref{defn:proc-provable}.
      \item "\obl{\D_0, \G_0 ||- e_0}" is provable
        \by \s31, Definition~\ref{defn:isabelle-embedding}.
      \item "\obl{\D_0 ||- f}" is provable
        \by provability of "\obl{\D_0, \obl{\G_0 ||- e_0} ||- f}", \s32, cut (Fact~\ref{thm:cut}).
      \item \Qed
        \by \s33, i.h. for "\EE_0"
      \end{ecom}

    \item \Case "\tau" is "\HIDE\ \Phi, \phi", \ie,
      \begin{gather*}
        \EE =
        \Im["\HIDE_1".]{"\sigma.\ \HIDE\ \Phi, \phi : \obl{\G_0, \phi, \G_1 ||- e} --> \obl{\D ||- f}"}
           {\deduce{"\sigma.\ \HIDE\ \Phi : \obl{\G_0, hide{\phi}, \G_1 ||- e} --> \obl{\D ||- f}"}{\EE_0}}
      \end{gather*}

      \begin{ecom}[{$\s3$}1.]
      \item "\obl{\G_0, \hide{\phi}, \G_1 ||- e}" is provable
        \by provability of "\obl{\D ||- f}", i.h. for "\EE_0".
      \item \Qed
        \by \s31, "\isa{\G_0, \hide{\phi}, \G_1 ||- e} = \isa{\G_0, \phi, \G_1 ||- e}" (Definition~\ref{defn:isabelle-embedding}).
      \end{ecom}

    \item \Case "\tau" is "\TAKE\ \nil", \ie,
      $
      \EE = 
      \Im["\TAKE_0".]{"\sigma.\ \TAKE\ \nil : \obl{\G ||- e} --> \obl{\G ||- e}"}{}
      $
      \Trivial

    \item \Case "\tau" is "\WITNESS\ \nil", \ie,
      $
      \EE = 
      \Im["\WITNESS_0".]{"\sigma.\ \WITNESS\ \nil : \obl{\G ||- e} --> \obl{\G ||- e}"}{}
      $
      \Trivial

    \item \Case "\tau" is "\TAKE\ u, \vec \beta", \ie,
      \begin{gather*}
        \EE =
        \Im["\TAKE_1".]{"\sigma.\ \TAKE\ u, \vec\beta : \obl{\G ||- \forall x : e} --> \obl{\D ||- f}"}
           {\deduce{"\sigma.\ \TAKE\ \vec\beta : \obl{\G, \NEW u ||- e [x := u]} --> \obl{\D ||- f}"}{\EE_0}}
      \end{gather*}

      \begin{ecom}[{$\s3$}1.]
      \item "\obl{\G, \NEW u ||- e [x := u]}" is provable
        \by i.h. for "\EE_0".
      \item \Qed
        \by \s31{} and predicate logic.
      \end{ecom}

    \item \Case "\tau" is "\sigma.\ \TAKE\ u \in T", \ie,
      \begin{gather*}
        \EE =
        \Im["\TAKE_2".]{"\sigma.\ \TAKE\ u \in T, \vec\beta : \obl{\G ||- \forall x \in S : e} --> \obl{\D ||- f}"}
           {"\obl{\G ||- S \subseteq T}"
            &
            \deduce{"\sigma.\ \TAKE\ \vec\beta : \obl{\G, \NEW u, u \in T ||- e [x := u]} --> \obl{\D ||- f}"}{\EE_0}}
      \end{gather*}

      \begin{ecom}[{$\s3$}1.]
      \item "\obl{\G, \NEW u, u \in T ||- e [x := u]}" is provable
        \by i.h on "\EE_0".
      \item "\obl{\G, \NEW u, u \in S ||- u \in T}" is provable
        \begin{ecom}[{$\s4$}1.]
        \item "\obl{\G, \NEW u ||- S \subseteq T}" is provable
          \by Definition~\ref{defn:proc-provable}, weakening
          (Fact~\ref{thm:weaken}).
        \item \Qed
          \by \s41, Definition of "\subseteq".
        \end{ecom}
      \item "\obl{\G, \NEW u, u \in S ||- e[x := u]}" is provable
        \by \s31, \s32, cut (Fact~\ref{thm:cut}).
      \item \Qed
        \by \s33{} and predicate logic.
      \end{ecom}

    \item \Case "\tau" is "\WITNESS\ w, \W", \ie,
      \begin{gather*}
        \EE =
        \Im["\WITNESS_1".]{"\sigma.\ \WITNESS\ w, \W : \obl{\G ||- \exists x : e} --> \obl{\D ||- f}"}
           {\deduce{"\sigma.\ \WITNESS\ \W : \obl{\G ||- e [x := w]} --> \obl{\D ||- f}"}{\EE_0}}
      \end{gather*}

      \begin{ecom}[{$\s3$}1.]
      \item "\obl{\G ||- e [x := w]}" is provable
        \by i.h. for "\EE_0".
      \item \Qed
        \by \s31.
      \end{ecom}

    \item \Case "\tau" is "\WITNESS\ w \in T, \W" and:
      \begin{gather*}
        \EE =
        \Im["\WITNESS_2".]
           {"\sigma.\ \WITNESS\ w \in T, \W : \obl{\G ||- \exists x \in S : e} --> \obl{\D ||- f}"}
           {"\obl{\G ||- T \subseteq S}"
            &
            "\obl{\G ||- w \in T}"
            &
            \deduce{"\sigma.\ \WITNESS\ \W : \obl{\G, w \in T ||- e [x := w]} --> \obl{\D ||- f}"}{\EE_0}}
     \end{gather*}

     \begin{ecom}[{$\s3$}1.]
     \item "\obl{\G, w \in T ||- e [x := w]}" is provable
       \by i.h. for "\EE_0".
     \item "\obl{\G ||- w \in T}" is provable
       \by Definition~\ref{defn:proc-provable}.
     \item "\obl{\G ||- e [x := w]}" is provable
       \by \s31, \s32, cut (Fact~\ref{thm:cut}).
     \item "\obl{\G ||- w \in S}" is provable
       \begin{ecom}[{$\s4$}1.]
       \item "\obl{\G, w \in T ||- w \in S}" is provable
         \by Definition~\ref{defn:proc-provable}, Definition of "\subseteq".
       \item \Qed
         \by \s41, \s32, cut (Fact~\ref{thm:cut}).
       \end{ecom}
     \item \Qed
       \by \s33, \s34, and predicate logic.
     \end{ecom}
     
    \item "\tau" is "\HAVE\ g", \ie,
      \begin{gather*}
        \EE =
        \Im[\HAVE.]{"\sigma.\ \HAVE\ g : \obl{\G ||- e => f} --> \obl{\G, g ||- f}"}
           {"\obl{\G, e ||- g}"}
      \end{gather*}

      \begin{ecom}[{$\s3$}1.]
      \item "\obl{\G, e, g ||- f}" is provable
        \by weakening (Fact~\ref{thm:weaken}).
      \item "\obl{\G, e ||- g}" is provable
        \by Definition~\ref{defn:proc-provable}.
      \item "\obl{\G, e ||- f}" is provable
        \by \s31, \s32, cut (Fact~\ref{thm:cut}).
      \item "\obl{\G ||- e => f}" is provable
        \by \s33.
      \end{ecom}

    \item "\sigma.\,\tau" is "\s n.\ \obl{\W ||- g}\ \PROOF\ \pi", \ie, 
      \begin{gather*}
        \EE =
        \Im[$\ASSERT_1$.]{"\s n.\ \obl{\W ||- g}\ \PROOF\ \pi : \obl{\G ||- e} --> \obl{\G, \obl{\W ||- g} ||- e}"}
           {\deduce{"\pi : \obl{\G, \hide{\lnot e}, \W ||- g}"}{\DD_0}}
      \end{gather*}

      \begin{ecom}[{$\s3$}1.]
      \item "\obl{\G, \hide{\lnot e}, \obl{\W ||- g} ||- e}" is provable
        \by weakening (Fact~\ref{thm:weaken}).
      \item "\obl{\G, \hide{\lnot e}, \W ||- g}" is provable
        \by i.h. for "\DD_0".
      \item "\obl{\G, \hide{\lnot e} ||- e}" is provable
        \by \s31, \s32, cut (Fact~\ref{thm:cut}).
      \item \Qed
        \by \s33, Fact~\ref{thm:classic}.
      \end{ecom}

    \item \Case "\sigma.\,\tau" is "\s n l.\ \obl{\W ||- g}\ \PROOF\ \pi", \ie,
      \begin{gather*}
        \EE =
        \Im[$\ASSERT_2$.]{"\s n l.\ \obl{\W ||- g}\ \PROOF\ \pi : \obl{\G ||- e} --> \obl{\G, \s n l \DEF \obl{\W ||- g}, \hide{\s n l} ||- e}"}
           {\deduce{"\pi : \obl{\G, \s n l \DEF \obl{\W ||- g}, \hide{\lnot e}, \W ||- g}"}{\DD_0}}
      \end{gather*}

      \begin{ecom}[{$\s3$}1.]
      \item "\obl{\G, \s n l \DEF \obl{\W ||- g}, \hide{\lnot e}, \hide{\s n l} ||- e}" is provable
        \\\mbox{}
        \by provability of "\obl{\G, \s n l \DEF \obl{\W ||- g}, \hide{\s n l} ||- e}", weakening (Fact~\ref{thm:weaken}).
      \item "\obl{\G, \s n l \DEF \obl{\W ||- g}, \hide{\lnot e}, \hide{\obl{\W ||- g}} ||- e}" is provable
        \by \s31, expansion (Fact~\ref{thm:expand}).
      \item "\obl{\G, \s n l \DEF \obl{\W ||- g}, \hide{\lnot e}, \W ||- g}" is provable
        \by i.h. for "\DD_0".%
      \item "\obl{\G, \s n l \DEF \obl{\W ||- g}, \hide{\lnot e} ||- e}" is provable
        \by \s32, \s33, cut (Fact~\ref{thm:cut}).
      \item "\obl{\G, \hide{\lnot e} ||- e}" is provable
        \by \s34, strengthening (Fact~\ref{thm:delete}).
      \item \Qed
        \by \s35, Fact~\ref{thm:classic}.
      \end{ecom}
      
    \item "\tau" is "\CASE\ g\ \PROOF\ \pi", \ie,
      \begin{gather*}
        \EE =
        \Im[\CASE.]{"\sigma.\ \CASE\ g\ \PROOF\ \pi : \obl{\G ||- e} --> \obl{\D ||- f}"}
           {\deduce{"\sigma.\ \obl{g ||- e}\ \PROOF\ \pi : \obl{\G ||- e} --> \obl{\D ||- f}"}{\EE_0}}
      \end{gather*}

      \textit{By} i.h. for "\EE_0".

    \item "\tau" is "\s n.\ \SUFFICES\ \obl{\W ||- g}\ \PROOF\ \pi", \ie,
      \begin{gather*}
        \EE =
        \Im[$\SUFFICES_1$.]{"\s n.\ \SUFFICES\ \obl{\D ||- f}\ \PROOF\ \pi : \obl{\G ||- e} --> \obl{\G, \hide{\lnot e}, \W ||- g}"}
           {\deduce{"\pi : \obl{\G, \obl{\W ||- g} ||- e}"}{\DD_0}}
      \end{gather*}

    \begin{ecom}[{$\s3$}1.]
      \item "\obl{\G, \hide{\lnot e}, \obl{\W ||- g} ||- e}" is provable
        \by i.h. for "\DD_0", weakening (Fact~\ref{thm:weaken}).
      \item "\obl{\G, \hide{\lnot e} ||- e}" is provable
        \by provability of "\obl{\G, \hide{\lnot e}, \W ||- g}", \s31, cut (Fact~\ref{thm:cut}).
      \item \Qed
        \by \s32, Fact~\ref{thm:classic}.
      \end{ecom}

    \item "\sigma.\,\tau" is "\s n l.\ \SUFFICES\ \obl{\W ||- g}\ \PROOF\ \pi", \ie,
      \begin{gather*}
        \EE =
        \Im[$\SUFFICES_2$.]{"\s n l.\ \SUFFICES\ \obl{\W ||- g}\ \PROOF\ \pi : \obl{\G ||- e} --> \obl{\G, \s n l \DEF \obl{\W ||- g}, \hide{\lnot e}, \W ||- g}"}
           {\deduce{"\pi : \obl{\G, \s n l \DEF \obl{\W ||- g}, \hide{\s n l} ||- e}"}{\DD_0}}
      \end{gather*}

      \begin{ecom}[{$\s3$}1.]
      \item "\obl{\G, \s n l \DEF \obl{\W ||- g}, \hide{\lnot e}, \hide{\s n l} ||- e}" is provable
        \by i.h. for "\DD_0", weakening (Fact~\ref{thm:weaken}).
      \item "\obl{\G, \s n l \DEF \obl{\W ||- g}, \hide{\lnot e}, \hide{\obl{\W ||- g}} ||- e}" is provable
        \by \s31, expansion (Fact~\ref{thm:expand}).
      \item "\obl{\G, \s n l \DEF \obl{\W ||- g}, \hide{\lnot e} ||- e}" is provable
        \\\mbox{}
        \by \s32, provability of "\obl{\G, \s n l \DEF \obl{\W ||- g}, \hide{\lnot e}, \W ||- g}", cut (Fact~\ref{thm:cut}).
      \item "\obl{\G, \hide{\lnot e} ||- e}" is provable
        \by \s33, strengthening (Fact~\ref{thm:delete}).
      \item \Qed
        \by \s34, Fact~\ref{thm:classic}.
      \end{ecom}

    \item \Case "\tau" is "\PICK\ \vec\beta : p\ \PROOF\ \pi", \ie,
      \begin{gather*}
        \EE =
        \Im[\PICK.]{"\sigma.\ \PICK\ \vec\beta : p\ \PROOF\ \pi :
                        \obl{\G ||- e} --> \obl{\G, \refl{\vec \beta}, p ||- e}"}
           {\deduce{"\pi : \obl{\G ||- \exists \vec \beta : p}"}{\DD_0}}
      \end{gather*}

      \begin{ecom}[{$\s3$}1.]
      \item "\obl{\G, \exists \vec \beta : p ||- e}" is provable
        \by provability of "\obl{\G, \refl{\vec \beta}, p ||- e}", predicate logic.
      \item "\obl{\G ||- \exists \vec \beta : p}" is provable
        \by i.h. for "\DD_0".
      \item \Qed
        \by \s31, \s32, cut (Fact~\ref{thm:cut}).
      \end{ecom}

    \item \Qed
      \by \s21, \dots, \s220
    \end{ecom}
  \item \Qed
    \by \s11, \s12.
  \end{ecom}
\end{proof}

\subsection{Constrained Search}
\label{apx:constraints}

The correctness theorem (\ref{thm:correctness}) establishes an
implication from the leaf obligations generated by a complete
meaningful claim to the obligation of the claim. It is always true,
regardless of the provability of any individual leaf obligation. While
changing the visibility of assumptions in an obligation does not
change its provability, a back-end prover may fail to prove it if
important assumptions are hidden.  As already mentioned in
Section~\ref{sec:obligations}, the \PM removes these hidden
assumptions before sending a leaf obligation to a back-end prover.
Therefore, in order to establish the Structural Soundness Theorem
(\ref{thm:meaning}), we must prove a property about the result of this
removal.

\begin{defn}[Filtration] \label{defn:filter} 
  The \emph{filtered} form of any obligation "\phi", written
  "\filter\phi", is obtained by deleting all assumptions of the form
  "\hide{\phi_0}" and replacing all assumptions of the form "\hide{o
    \DEF \delta}" with "\NEW o" anywhere inside "\phi".
\end{defn}

\noindent%
For example, "\filter{\NEW x, \hide{y \DEF x} ||- x = y} = \obl{\NEW
  x, \NEW y ||- x = y}". We thus see that filtration can render a true
obligation false; however, if the filtered form of an obligation is
true, then so is the obligation.

\begin{lem}[Verification Lemma] \label{thm:verification}
  If "\filter{\phi}" is provable, then "\phi" is provable.
\end{lem}

\begin{proof}[Proof Sketch]
  By induction on the structure of the obligation "\phi", with each
  case a straightforward consequence of facts~\ref{thm:definition}
  and~\ref{thm:weaken}.
\end{proof}

\begin{defn}[Verifiability] \label{defn:verifiability}
  The obligation "\phi" is said to be  \emph{verifiable} if "\filter\phi" is
  provable.
\end{defn}

\noindent%
We now prove the Structural Soundness Theorem (\ref{thm:meaning}).

\setcounter{thm}{0}

\begin{thm}
  If "\pi:\phi" is a complete meaningful claim and every leaf
  obligations it generates is verifiable, then "\phi" is true.
\end{thm}

\begin{proof}\mbox {} \small
  \begin{ecom}[{$\s1$}1.]
  \item For every leaf obligation "\phi_0" generated by "\pi : \phi",
    it must be that "\phi_0" is provable.
    \begin{ecom}[{$\s2$}1.]
    \item \textit{Take} "\phi_0" as a leaf obligation generated by
      "\pi : \phi".
    \item "\filter{\phi_0}" is provable \by assumption and Definition~\ref{defn:verifiability}.
    \item \Qed \by \s22, Verification Lemma~\ref{thm:verification}.
    \end{ecom}
  \item "\phi" is provable \by \s11, Correctness Theorem~\ref{thm:correctness}.
  \item \Qed \by \s12, Trust Axiom~\ref{axm:trust}.
  \end{ecom}
\end{proof}

\clearpage

\section{A \tlatwo Proof of Cantor's Theorem}
\label{apx:cantor}

The following is the complete \tlatwo proof of Cantor's theorem
referenced in Section~\ref{sec:proof-language.lang}.
\begin{quote} \small
  \begin{tabbing}
    \THEOREM\ "\forall S : \forall f \in [S -> \SUBSET\ S] : \exists A \in \SUBSET\ S : \forall x \in S : f[x] \neq A" \\
    \s1m.\ \= \kill
    \s11.\ \> \ASSUME \= "\NEW\ S", \\
           \>         \> "\NEW\ f \in [S -> \SUBSET\ S]" \\
           \> \PROVE "\exists A \in \SUBSET\ S : \forall x \in S : f[x] \neq A" \\
           \hspace{1em}\= \s2m.\ \= \kill
           \> \s21.\ \> \DEFINE "T \DEF \{z \in S : z \notin f[z]\}" \\
           \> \s22.  \> "\forall x \in S : f[x] \neq T" \\
           \>        \hspace{1em}\= \s3m.\ \= \kill
           \>        \> \s31.\ \> \ASSUME "\NEW\ x \in S" \PROVE "f[x] \neq T" \\
           \>        \>        \hspace{1em}\=\s4m.\ \= \kill
           \>        \>        \> \s41.\ \> \CASE "x \in T" \OBVIOUS \\
           \>        \>        \> \s42.\ \> \CASE "x \notin T" \OBVIOUS \\
           \>        \>        \> \s43.\ \> \QED \BY \s41, \s42 \\
           \>        \hspace{1em}\= \s3m.\ \= \kill
           \>        \> \s32.  \> \QED\ \BY\ \s31 \\
           \hspace{1em}\= \s2m.\ \= \kill
           \> \s23.  \> \QED\ \BY\ \s22 \\
    \s1m.\ \= \kill
    \s12.  \> \QED\ \BY\ \s11
  \end{tabbing}
\end{quote}
As an example, the leaf obligation generated (see
Appendix~\ref{apx:correctness}) for the proof of \s41 is:
\def\ss#1#2{\hbox{\s{#1}{#2}}}
\begin{quote} \small
  \begin{tabbing}
    "\smash{\Bigl(}"\ \ 
    \= "\ss11 \DEF \obl{\NEW S, \NEW f, f \in [S -> \SUBSET\ S] ||- \exists A \in \SUBSET\ S : \forall x \in S : f[x] \neq A}", \\
    \> "\NEW S", \\
    \> "\NEW f", "f \in [S -> \SUBSET\ S]", \\
    \> "T \DEF \{z \in S : z \notin f[z]\}", \\
    \> "\hide{\lnot \left(\exists A \in \SUBSET\ S : \forall x \in S : f[x] \neq A\right)}", \\
    \> "\ss22 \DEF \forall x \in S : f[x] \neq T", \\
    \> "\hide{\lnot \left( \forall x \in S : f[x] \neq T \right)}", \\
    \> "\ss31 \DEF \obl{\NEW x, x \in S ||- f[x] \neq T}", \\
    \> "\NEW x", "x \in S", \\
    \> "\hide{\lnot \left( f[x] \neq T \right)}", \\
    \> "\ss41 \DEF \obl{x \in T ||- f[x] \neq T}", \\
    \> "x \in T" \\
    "||-" \> "f[x] \neq T" "\smash{\Bigr)}".
  \end{tabbing}
\end{quote}
Filtering its obligation (see Definition~\ref{defn:filter}) and expanding all definitions gives:
\begin{quote} \small
  \begin{tabbing}
    "\smash{\Bigl(}"
    \= "\NEW S", \\
    \> "\NEW f", "f \in [S -> \SUBSET\ S]", \\
    \> "\NEW x", "x \in S", \\
    \> "x \in \{z \in S : z \notin f[z]\} ||- f[x] \neq \{z \in S : z \notin f[z]\}" "\smash{\Bigr)}".
  \end{tabbing}
\end{quote}
In Isabelle/\tlaplus, this is the following lemma:
\begin{quote} \small
  \begin{tabbing}
    \texttt{lemma} \= "!! S." \\
                   \> "!! f.\ " \= " f \in [S -> \SUBSET\ S] ==>" \\
                   \>                 \> "\Bigl(!! x.\ \bigl\llbracket" \= "x \in S ;" \\
                   \> \> \> "x \in \{z \in S : z \notin f[z]\} \ \bigr\rrbracket
                         ==> f[x] \neq \{z \in S : z \notin f[z]\} \Bigr)"
  \end{tabbing}
\end{quote}

\end{document}